\documentclass[runningheads]{llncs}
\usepackage[T1]{fontenc}
\usepackage{graphicx}

\usepackage{thm-restate}
\usepackage{multirow}
\usepackage{amsmath}
\usepackage{amssymb}
\usepackage{algorithm}
\usepackage[noend]{algorithmic}
\usepackage{cleveref}

\renewcommand{\emptyset}{\varnothing}

\newcommand{\E}{\mathbb{E}}
\newcommand{\N}{\mathbb{N}}
\newcommand{\Su}{S^{\uparrow}}
\newcommand{\Sd}{S^{\downarrow}}
\newcommand{\cD}{\mathcal{D}}

\newcommand{\ind}{\mathbf{1}}
\newcommand{\Alg}{\mathsf{Alg}}
\newcommand{\Sel}{\textsc{Sel}}
\newcommand{\LSort}{\textsc{Quant}}

\newcommand{\argmax}{\mathrm{argmax}}

\newcommand{\cR}{\mathcal{R}}

\newcommand{\PP}{\mathbb{P}}
\newcommand{\EE}{\mathbb{E}}

\newcommand{\Xc}{\mathcal{X}}

\newcommand{\Ac}{\mathcal{A}}
\newcommand{\Actil}{\widetilde{\mathcal{A}}}
\newcommand{\Ec}{\mathcal{E}}

\newcommand{\Jc}{\mathcal{J}}
\newcommand{\Hc}{\mathcal{H}}
\newcommand{\Kc}{\mathcal{K}}
\renewcommand{\Pr}{\PP}
\newcommand{\Mod}[1]{\ (\mathrm{mod}\ #1)}

\allowdisplaybreaks

\begin{document}
\title{Complexity of Round-Robin Allocation with Potentially Noisy Queries}
%
%
\author{
Zihan Li\inst{1} \and
Pasin Manurangsi\inst{2} \and
Jonathan Scarlett\inst{1} \and
Warut Suksompong\inst{1}
}
\authorrunning{Z. Li et al.}
%
\institute{National University of Singapore, Singapore
\and
Google Research, Thailand
}
\maketitle              
\begin{abstract}
We study the complexity of a fundamental algorithm for fairly allocating indivisible items, the round-robin algorithm.
For $n$ agents and $m$ items, we show that the algorithm can be implemented in time $O(nm\log(m/n))$ in the worst case.
If the agents' preferences are uniformly random, we establish an improved (expected) running time of $O(nm + m\log m)$.
On the other hand, assuming comparison queries between items, we prove that $\Omega(nm + m\log m)$ queries are necessary to implement the algorithm, even when randomization is allowed.
We also derive bounds in noise models where the answers to queries are incorrect with some probability.
Our proofs involve novel applications of tools from multi-armed bandit, information theory, as well as posets and linear extensions.
\end{abstract}

\section{Introduction}

A famous computer science professor is retiring soon, and she wants to distribute the dozens of books that she has accumulated in her office over the years as a parting gift to her students.
As one may expect, the students have varying preferences over the books, depending on their favorite authors or the branches of computer science that they specialize in.
How can the professor take the students' preferences into account and distribute the books in a fair manner?

The problem of fairly allocating scarce resources has long been studied under the name of \emph{fair division} \cite{BramsTa96,Steinhaus48} and received significant interest in recent years \cite{AmanatidisAzBi23,Aziz20,Moulin19,Suksompong21,Walsh20}.
A simple and well-known procedure for allocating discrete items---such as books, clothes, or household items---is the \emph{round-robin algorithm}.
In this algorithm, the agents take turns picking their most preferred item from the remaining items according to a cyclic agent ordering, until all items have been allocated.
The algorithm is sometimes known as the \emph{draft} mechanism for its use in allocating sports players to teams \cite{CaragiannisKuMo19}.
Despite its simplicity, the allocation chosen by the round-robin algorithm satisfies a surprisingly strong fairness guarantee called \emph{envy-freeness up to one item (EF1)} provided that the agents have additive utilities over the items. 
EF1 means that if an agent envies another agent, this envy can be eliminated by removing some item from the latter agent's bundle.
Furthermore, round-robin can be implemented using only the agents' \emph{ordinal} rankings over individual items.
This stands in stark contrast to other important fair division algorithms such as \emph{envy-cycle elimination}, which requires eliciting the agents' rankings over \emph{sets} of items, or \emph{maximum Nash welfare}, for which the agents' \emph{cardinal} utilities for items must be known.\footnote{For descriptions of these algorithms, we refer to the survey by Amanatidis et al.~\cite{AmanatidisAzBi23}.}

It is clear that the round-robin algorithm can be implemented in time polynomial in the number of agents and items.
However, despite being one of the very few basic algorithms in discrete fair division and used, adapted, and extended numerous times to provide various fairness guarantees, its complexity has not been analyzed in detail to our knowledge.
A moment of thought reveals two sensible approaches for implementing round-robin.\footnote{We remain informal with the query model for this discussion, but will make this precise later.}
Firstly, we can find each agent's ranking over all individual items---this allows us to determine the item picked in every turn, no matter who the picker is.
Since sorting $m$ numbers can be done in $O(m\log m)$ time, this approach takes $O(nm\log m)$ time, where $n$ and $m$ denote the number of agents and items, respectively.
Secondly, we can instead, for the picking agent at each turn, find the agent's most valuable item from the remaining items.
As there are $m$ turns and finding the maximum among $m$ numbers takes $O(m)$ time, the time complexity of this approach is $O(m^2)$.
Hence, the first approach is better when $n$ grows at a slower rate than $m/\log m$, while the second approach is more efficient if $n$ has a higher growth rate.
Are these two approaches already the best possible,
or are there faster ways---that is, faster than $O(m\cdot\min\{n\log m, m\})$ time---to implement the fundamental round-robin algorithm?

\renewcommand{\arraystretch}{1.2}
\begin{table*}[!t]
\centering
\begin{tabular}{|c|c|c|c|c|}
\hline
\textbf{Noise} & \textbf{Preferences} & \textbf{Queries} & \textbf{Upper Bound} & \textbf{Lower Bound} \\ \hline\hline
\multirow{4}*{Noiseless} & \multirow{2}*{Worst-case} & Comparison & $O(nm\log(m/n))$ & $\Omega(nm+m\log m)$ \\ \cline{3-5}
& & Value & $O(nm)$ & $\Omega(nm)$ \\ \cline{2-5}
& \multirow{2}*{Uniformly random} & Comparison & $O(nm + m\log m)$ & $\Omega(nm)$ \\ \cline{3-5}
&  & Value & $O(nm)$ & $\Omega(nm)$ \\ \hline \hline
\multirow{4}*{Noisy} & \multirow{4}*{Worst-case} & \multirow{2}*{Comparison} & \multirow{2}*{$O(nm\log(m/\delta))$} & $\Omega(nm\log(1/\delta)$ \\
& & & & $+ m\log(m/\delta))$ \\ 
\cline{3-5}
&  & \multirow{2}*{Value} & \multirow{2}*{$O(nm\log(m/\delta))$} & $\Omega(nm\log(1/\delta)$ \\
& & & & $+ m\log(m/\delta))$ \\ \cline{1-5}
\end{tabular}

\vspace{3mm}

\caption{Summary of our results on the query complexity of round-robin allocation.
In the noisy setting, $\delta$ is the allowed error probability.
The uniformly random lower bounds hold because the proof of \Cref{thm:noiseless-comparison-lower-nm} (and \Cref{cor:noiseless-value-lower-nm}) uses uniformly random preferences.
Our algorithms also offer running time guarantees that match the comparison query upper bounds.}
\label{table:summary}
\end{table*}

\subsection{Overview of Results}

Fix the agent ordering $1,2,\dots,n$, and assume that each agent has a strict ranking over the $m$ items, where $m\ge n$.
Hence, the round-robin allocation is uniquely defined, and the task of an algorithm is to output this allocation.
Note that we do not require the algorithm to output the item that each agent picks in each turn---this makes our lower bounds stronger, while for the upper bounds, our algorithms can also return this additional information.
We consider two models for eliciting agent preferences.
In the \emph{comparison query} model, an algorithm can find out with each query which item an agent prefers between a pair of items; in the \emph{value query} model, it can find out the utility of an agent for an item.

In \Cref{sec:noiseless}, we consider the \emph{noiseless setting}, where the answers to queries are always accurate.
We present a deterministic algorithm that runs in time $O(nm\log(m/n))$ in both query models.
Since $nm\log(m/n)$ is less than both $nm\log m$ and $m^2$, our algorithm is more efficient than both of the approaches mentioned earlier, for any asymptotic relation between $n$ and $m$.
We then show that when the preferences are uniformly random---meaning that each agent has a uniformly random and independent ranking of the items---we can obtain an improved (expected) running time of $O(nm + m\log m)$.
We complement these results by establishing lower bounds: $\Omega(nm + m\log m)$ and $\Omega(nm)$ queries are necessary in the comparison and value model, respectively.
Since the entire preferences can be elicited via $O(nm)$ queries in the value model, the latter bound is tight.
Our lower bounds hold even against algorithms that may fail with some constant probability; proving the former bound entails leveraging results on posets and linear extensions.

In \Cref{sec:noisy}, we turn our attention to the \emph{noisy setting}.
For comparison queries, we assume that the answer to each query is incorrect with probability $\rho$ independently of all other queries, where $\rho \in (0, 1/2)$ is a given constant.
For value queries, the answer to each query is the true utility with probability $1-\rho$ and an arbitrary value with probability $\rho$, where this value can be chosen adversarially.
We focus on algorithms that are correct with probability at least $1-\delta$ for a given parameter\footnote{We assume that $\delta\in(0, 1/2-c)$ for some constant $c > 0$.} $\delta$.
We show that for both types of queries, there exists a deterministic algorithm running in time $O(nm\log(m/\delta))$.
On the other hand, we provide a lower bound of $\Omega(nm\log(1/\delta) + m\log(m/\delta))$ on the number of queries even for randomized algorithms;\footnote{Note that if $\delta\in O(m^{-d})$ for some constant $d > 0$, the upper and lower bounds asymptotically match.} the proof involves novel applications of tools from multi-armed bandit and information theory and may be of independent interest to researchers in different areas.

A summary of our results can be found in \Cref{table:summary}.

\subsection{Related Work}

The fair division literature typically considers fairness guarantees (such as EF1) that are feasible in different settings along with algorithms that achieve these guarantees \cite{AmanatidisAzBi23,BramsTa96,Moulin19}.
We take a central algorithm in the literature---the round-robin algorithm---and analyze its complexity.
Query complexity in discrete fair division has previously been studied by Plaut and Roughgarden~\cite{PlautRo20} and Oh et al.~\cite{OhPrSu21}.
For instance, Oh et al.~showed that for two agents, it is possible to compute an EF1 allocation using $O(\log m)$ queries.
However, all of these authors assumed a query model such that with each query, an algorithm can find out an agent's utility for any \emph{set} of items.
Since determining this value for a large set can be quite demanding, our query models, in which each query only involves one or two items, are arguably more realistic.

The vast majority of work in fair division assumes that accurate information on the agent preferences is available to the algorithms.
However, in reality there may be noise or uncertainty in these preferences, possibly due to the limited time or high cost for determining the true preferences.
Aziz et al.~\cite{AzizBiHa19} and Li et al.~\cite{LiBeYa22} investigated uncertainty models for item allocation.
For example, in their ``compact indifference model'', each agent reports a ranking over items that may contain ties; the ties indicate that the agent is uncertain about her preferences among the tied items.
Note that this approach to uncertainty is very different from ours; in our approach, the uncertainty does not appear in the initial input to the problem, but instead arises as noise in answers to queries.
Our noisy comparison model has been used for several algorithmic problems including searching, sorting, and selection \cite{BenorHa08,BravermanMoWe16,FeigeRaPe94,GanWiZh22,GuXu23}, and our noisy value model has also been studied for similar problems \cite{CohenaddadMaMa20}, though the use of these models in fair division is new to the best of our knowledge.

Not surprisingly given its wide applicability, the round-robin algorithm has been examined from various angles, including strategic considerations \cite{AzizBoLa17,BouveretLa14}, equilibrium properties \cite{AmanatidisBiLa23}, and monotonicity guarantees \cite{ChakrabortyScSu21}.

\section{Preliminaries}

Let $N = [n]$ be the set of agents and $M = [m]$ be the set of items, where $[k] := \{1,2,\dots,k\}$ for any positive integer~$k$.
Denote by $u_i(j)\ge 0$ the \emph{utility} (also referred to as the \emph{value}) of agent~$i$ for item~$j$, and assume for convenience that $u_i(j)\ne u_i(j')$ for all $i\in N$ and distinct $j,j'\in M$.
For $X, Y\subseteq M$, we write $X\succ_i Y$ if $u_i(x) > u_i(y)$ for all $x\in X$ and $y\in Y$.
The \emph{round-robin allocation} is the allocation that results from letting agents take turns picking items in the order $1,2,\dots,n,1,2,\dots$, where in each turn, the picking agent picks the item for which she has the highest utility.
Since each agent has distinct utilities for all items, this allocation is uniquely defined.
We sometimes refer to each sequence $1,2,\dots,n$ as a \emph{round}, where the last round might not include all agents.
Assume without loss of generality that $m\ge n$ (otherwise, we may simply ignore the agents who do not receive any item) and $n\ge 2$.

We consider two models of how an algorithm can discover information about agents' utilities.
In the \emph{comparison query} model, an algorithm can specify an agent $i\in N$ and a pair of distinct items $j,j'\in M$, and find out whether agent~$i$ prefers item $j$ to $j'$.
In the \emph{value query} model, an algorithm can specify $i\in N$ and $j\in M$, and find out the value of $u_i(j)$.
We assume that each query takes constant time and the algorithm can be adaptive.
In the \emph{noiseless setting} (\Cref{sec:noiseless}), the answer to each query is always accurate.
Observe that with no noise, each comparison query can be simulated using two value queries, so a lower bound for value queries implies a corresponding one for comparison queries with an extra factor of $1/2$.
In the \emph{noisy setting} (\Cref{sec:noisy}), for comparison queries, there is a constant $\rho\in(0,1/2)$, and the answer to each query is incorrect with probability $\rho$, independently of all other queries (including those with the same $i,j,j'$).\footnote{While one could consider an alternative model in which comparison faults are \emph{persistent}, there is no hope of getting a reasonable success rate for our problem under that model. 
For example, even with \emph{one} persistent error, if that error is on agent~$1$’s comparison between the top two items and some other agent has the same favorite item as agent~$1$, then the output will be wrong.}
For value queries, the answer to each query is the true utility with probability $1 - \rho$ and an arbitrary value with probability $\rho$, independently of all other queries (including those with the same $i,j$).
The arbitrary value can be chosen by an adversary, who may choose the value based not only on the algorithm's past queries (and their answers), but also on the entire set of utilities $u_i(j)$ for $i\in N$ and $j\in M$.

By \emph{uniformly random preferences}, we refer to the setting where we associate each agent~$i$ with a permutation $\sigma_i: [m] \to [m]$ chosen uniformly and independently at random, and let $i$ prefer $j$ to $j'$ if and only if $\sigma_i(j) > \sigma_i(j')$.

\subsection{Selection \& Quantiles Algorithms}

\paragraph{Selection}
We will use the following classic linear-time algorithm for the so-called \emph{selection} problem. 

\begin{lemma}[Selection algorithm~\cite{BlumFlPr73}] \label{lem:selection}
For every $i \in N$, $S \subseteq M$, and $\ell \in [|S|]$, there is an $O(|S|)$-time deterministic algorithm $\textsc{Sel}_{i, \ell}(S)$ that makes $O(
|S|)$ comparison queries and outputs a partition $(\Su, \Sd)$ of $S$ such that $\Su \succ_i \Sd$ and $|\Su| = \ell$.
\end{lemma}

\paragraph{Quantiles}
We next consider the \emph{$(m,n)$-quantiles} problem, where we want to partition a set of $m$ items into subsets of size at most $n$ each, so that every item in the first set is preferred to every item in the second set, which is in turn preferred to every item in the third set, and so on. 
This problem can be solved in $O(m \log(m/n))$ time (in contrast to $O(m \log m)$ for sorting).

\begin{lemma}[Quantiles algorithm] \label{lem:partial-sort}
For every $i \in N$, there is an $O(m \log (m/n))$-time deterministic algorithm $\LSort_i$ that makes $O(m \log (m/n))$ comparison queries and outputs a partition $(S^i_1, \dots, S^i_k)$ of $M$ for some $k\in\N$, with the property that $S^i_1 \succ_i \dots \succ_i S^i_k$ and $|S^i_1|, \dots, |S^i_k| \leq n$.
\end{lemma}

Similar results are known in the literature; see, e.g., Exercise 9.3-6 of Cormen et al.~\cite{CLRS-book}.
For completeness, we provide the proof of \Cref{lem:partial-sort}. 

\begin{proof}
The algorithm $\LSort_i$ on input $S \subseteq M$ works as follows:
\begin{itemize}
\item If $|S| \leq n$, terminate.
\item Otherwise, let $(\Su, \Sd) \gets \Sel_{i, \lfloor |S| /2\rfloor}$ (where $\Sel$ is from \Cref{lem:selection}). Then, recurse on $\Su$ and $\Sd$.
\end{itemize}
We start with $S = M$. 
Observe that the recurrence relation for both the running time and the query complexity is
\begin{align*}
T(\ell) \leq
\begin{cases}
O(1) &\text{ if } \ell \leq n; \\
O(\ell) + T(\lfloor \ell / 2 \rfloor) + T(\lceil \ell / 2 \rceil) &\text{ otherwise,}
\end{cases}
\end{align*}
where $\ell$ denotes the size of the set $S$. 
One can verify using standard methods for solving recurrence relations (e.g., a recursion tree) that $T(\ell) \le O(\ell \log(\ell/n))$. $\hfill \square$
\end{proof}

\section{Noiseless Setting}
\label{sec:noiseless}

We begin in this section by considering the noiseless setting, where every query always receives an accurate answer.

\subsection{Upper Bounds}
\label{sec:noiseless-upper}

First, we show that it is possible to achieve a running time of $O(nm \log(m/n))$ for comparison queries.

\begin{theorem}
\label{thm:noiseless-comparison-upper}
Under the noiseless comparison query model, there exists a deterministic algorithm that outputs the round-robin allocation using $O(nm\log(m/n))$ queries and $O(nm\log(m/n))$ time.
\end{theorem}

\begin{algorithm}[tb]
\caption{For worst-case preferences}
\label{alg:worst-case}
\textbf{Input}: Set of agents $N$, set of items $M$, utilities $u_i(j)$ for $i\in N, j\in M$ \\
\textbf{Output}: Round-robin allocation $(A_1,\dots,A_n)$

\begin{algorithmic}[1] 
\FOR{$i\in N$}
\STATE $A_i\leftarrow\emptyset$
\STATE $(S^i_1, \dots, S^i_k) \leftarrow \LSort_i(M)$ \hfill \COMMENT{See \Cref{lem:partial-sort}}
\STATE $b_i \leftarrow 1$ \hfill \COMMENT{First $b$ such that $S^i_b \ne \emptyset$} 
\ENDFOR
\STATE $S\leftarrow M$ \hfill \COMMENT{Set of remaining items}
\FOR{$r = 1,\dots,\lceil m/n \rceil$}
\FOR{$i = 1, \dots, \min\{n, m - n(r - 1)\}$}
\WHILE{$S^i_{b_i} \cap S = \emptyset$} \label{line:if-emptybucket}
\STATE $b_i \gets b_i + 1$
\ENDWHILE
\STATE $j\leftarrow $ best item for agent~$i$ in $S^i_{b_i} \cap S$ \label{line:find-best}
\STATE $A_i\leftarrow A_i\cup\{j\}$
\STATE $S\leftarrow S\setminus\{j\}$
\ENDFOR
\ENDFOR
\RETURN{$(A_1,\dots,A_n)$}
\end{algorithmic}
\end{algorithm}

\begin{proof}
The algorithm is presented as \Cref{alg:worst-case}.
Its correctness is due to the guarantee of \Cref{lem:partial-sort} that $S^i_1 \succ_i \cdots \succ_i S^i_k$, which implies that in each turn, the picking agent picks her most preferred item among the remaining items.

As for the number of comparison queries, note that there are only two places that require comparisons: (i) when we call $\LSort_i$, and (ii) when we find the best item for $i$ in $S^i_{b_i} \cap S$ (Line~\ref{line:find-best}). 
For (i), \Cref{lem:partial-sort} ensures that the number of queries is $O(m \log (m/n))$ for each $i$, resulting in a total of $O(nm \log(m/n))$ across all $i \in N$. 
For (ii), since $|S^i_{b_i}| \leq n$, we can find the most preferred item of agent~$i$ in $S^i_{b_i} \cap S$ using $O(n)$ queries.
Since Line~\ref{line:find-best} is invoked $m$ times, the number of queries for this part is $O(nm)$.
It follows that the total number of queries used by the algorithm is $O(nm \log(m/n))$.

Apart from Line~\ref{line:if-emptybucket}, it is clear that the running time of the rest of the algorithm is $O(nm \log(m/n))$. 
As for Line~\ref{line:if-emptybucket}, let us fix $i \in N$. 
Note that we can check whether $S^i_{b_i} \cap S$ is non-empty in time $O(|S^i_{b_i}|) \leq O(n)$. 
Since we increment $b_i$ each time the check fails and in each round the check passes only once, the total running time of this step (for this agent $i$) is at most $O\bigl(\sum_{b \in [k]} |S^i_b| + \lceil m/n \rceil \cdot n\bigr) = O(m)$.
Therefore, in total, the running time of this step across all agents is $O(nm)$. 
We conclude that the total running time of the entire algorithm is $O(nm \log(m/n))$, as desired. $\hfill \square$
\end{proof}

For value queries, we can query all $nm$ values and run the algorithm from \Cref{thm:noiseless-comparison-upper}.

\begin{corollary}
\label{cor:noiseless-value-upper}
Under the noiseless value query model, there is a deterministic algorithm that outputs the round-robin allocation using $O(nm)$ queries and $O(nm\log(m/n))$ time.    
\end{corollary}

Next, we consider uniformly random preferences, which constitute a standard stochastic model in fair division \cite{DickersonGoKa14,KurokawaPrWa16,ManurangsiSu20,ManurangsiSu21}. 
For these preferences, we present an improvement over the worst-case bound.
We remark that the round-robin algorithm on random preferences has been studied by Manurangsi and Suksompong~\cite[Thm.~3.1]{ManurangsiSu21} and Bai and G\"{o}lz~\cite[Prop.~2]{BaiGo22}.
Since both of these papers used preference models that imply uniformly random ordinal preferences,\footnote{
Specifically, Manurangsi and Suksompong~\cite{ManurangsiSu21} assumed that all agents' utilities for all items are drawn independently from the same (non-atomic) distribution, while Bai and G\"{o}lz~\cite{BaiGo22} allowed each agent's utilities for items to be drawn independently from an agent-specific (non-atomic) distribution.
} 
our result can be applied for the running time analysis of their algorithms.

\begin{theorem}
\label{thm:noiseless-comparison-upper-random}
For uniformly random preferences, under the noiseless comparison query model, there exists a deterministic algorithm that outputs the round-robin allocation using expected $O(nm + m \log m)$ queries and expected $O(nm + m \log m)$ time.
\end{theorem}

\begin{algorithm}[tb]
\caption{For uniformly random preferences}
\label{alg:uniformly-random}
\textbf{Input}: Set of agents $N$, set of items $M$, utilities $u_i(j)$ for $i\in N, j\in M$ \\
\textbf{Output}: Round-robin allocation $(A_1,\dots,A_n)$

\begin{algorithmic}[1] 
\FOR{$i\in N$}
\STATE $A_i\leftarrow\emptyset$
\STATE $L_i\leftarrow\emptyset$ \hfill \COMMENT{Sorted list of $i$'s best remaining items}
\ENDFOR
\STATE $S\leftarrow M$ \hfill \COMMENT{Set of remaining items}
\FOR{$r = 1,\dots,\lceil m/n \rceil$}
\FOR{$i = 1, \dots, \min\{n, m - n(r - 1)\}$}
\IF{$L_i = \emptyset$} \label{line:if-emptyset}
\STATE $\ell\leftarrow \lceil|S|/n\rceil$
\STATE $(\Su, \Sd) \gets \textsc{Sel}_{i, \ell}(S)$ \hfill \COMMENT{See \Cref{lem:selection}} 
\STATE Find the sorted list $L_i$ of $\Su$ for agent~$i$, using a standard sorting algorithm (e.g., merge sort).
\ENDIF \label{line:endif-emptyset}
\STATE $j\leftarrow $ best item for agent~$i$ in $L_i$
\STATE $A_i\leftarrow A_i\cup\{j\}$
\STATE $S\leftarrow S\setminus\{j\}$
\FOR{$i'\in N$}
\STATE $L_{i'}\leftarrow L_{i'}\setminus\{j\}$
\ENDFOR
\ENDFOR
\ENDFOR
\RETURN{$(A_1,\dots,A_n)$}
\end{algorithmic}
\end{algorithm}

\begin{proof}
The algorithm is shown as \Cref{alg:uniformly-random}.
The correctness of the algorithm is again trivial.
Furthermore, the running time does not add extra asymptotic terms on top of the query complexity, so we will only establish the latter. 
In particular, we will show that the expected number of queries made by the first agent is at most $O\bigl(m + \frac{m}{n} \log m\bigr)$; the proof is similar for the other agents, and the desired statement follows from summing this up across all agents.

Let $Q_n(m)$ denote the number of queries made by the first agent when there are $m$ items and $n$ agents.
Fix $n\in\mathbb{N}$.
We will show by induction on $m$ that
\begin{align} \label{eq:induction-claim}
Q_n(m) \leq C \cdot \left(m + \left\lceil \frac{m}{n} \right\rceil \log m\right) 
\end{align}
where $C > 0$ is a sufficiently large constant. 
Specifically, let $C_2$ be a constant such that the re-initialization of $L_i$ between Lines \ref{line:if-emptyset} and \ref{line:endif-emptyset} of the algorithm takes at most $C_2 \cdot (|S| + \ell \log \ell)$ comparison queries. Then, we let $C = 100 C_2$.

For the base case where $m \leq 2n$, note that there are at most two rounds and each round only takes at most $C_2 \cdot (m + \lceil m/n\rceil \log m)$ comparisons for the first agent.

Next, we address the induction step. 
Suppose that for some $m^* > 2n$, inequality~\eqref{eq:induction-claim} holds for all $m < m^*$; we will show that it also holds for $m = m^*$.

Consider running the algorithm for $m = m^*$. 
Let $R^* = \lceil m^* / n \rceil\ge 3$ be the total number of rounds to be run. 
Let $r \geq 1$ denote the first round such that $L_1$ becomes empty after the end of the round; for notational convenience, we let $r = m^* / n$ instead of $\lceil m^* / n\rceil$ in the case that this happens in the last round.
Note that $r$ is a random variable. 
Observe that from round $r + 1$ onward, the expected number of queries made by the agent is the same as if the algorithm is run on $m^* - rn$ items.\footnote{This is because, conditioned on the items selected so far by all agents, the remaining items admit uniformly random preferences.} 
In other words, we have
\begin{align*}
Q_n(m^*) &\le C_2 \cdot (m^* + R^* \log m^*) + \E_r[Q_n(m^* - rn)],  
\end{align*}
where we use the convention $Q_n(0) = 0$.

Plugging the inductive hypothesis into the inequality above, we get that $Q_n(m^*)$ is at most 
\begin{align*}
&C_2 \cdot (m^* + R^* \log m^*) 
+ \E_r\left[C\cdot\left((m^* - rn) + (R^* - r) \log m^*\right)\right] \\
&= C \cdot \left(m^* + R^* \log m^*\right) 
+C_2\cdot \left((m^* - 100 n X) + (R^* - 100 X) \log m^*\right),
\end{align*}
where $X := \E[r]$.
Hence, to show \eqref{eq:induction-claim} for $m = m^*$, it suffices to prove that $X \geq R^* / 100$. 
In turn, to prove this, it is sufficient to show that $\Pr[r \geq R^*/50] \geq 0.5$, or equivalently, $\Pr[r < R^*/50] \leq 0.5$.

To show that $\Pr[r < R^*/50] \leq 0.5$, let $r_0 := \lfloor R^* / 50 \rfloor$. 
Let $X_{i',r'}$ be an indicator variable that equals $1$ if and only if \emph{both} of the following are true: (i) agent $i'$ takes an item from $L_1$ in round $r'$, and (ii) $L_1$ has not been re-initialized by round $r'$. 
Notice that
\begin{align*}
\Pr\left[r < \frac{R^*}{50}\right] 
&= \Pr\left[R^* = \sum_{r' \in [r_0]} \sum_{i' \in N} X_{i', r'} \right] 
\leq \frac{1}{R^*}\cdot\E\left[\sum_{r' \in [r_0]} \sum_{i' \in N} X_{i', r'} \right],
\end{align*}
where the equality follows from the fact that we start with $L_1$ of size $R^*$, and the inequality follows from Markov's inequality. Thus, to show that $\Pr[r < R^*/50] \leq 0.5$, it suffices to show that $\E\left[\sum_{r' \in [r_0]} \sum_{i' \in N} X_{i', r'} \right] \leq 0.5R^*$.

To calculate this expectation, let us make the following observations on $X_{i', r'}$.
If $i' \ne 1$, then since each agent $i'$ removes an item from $S$ most preferred by her at that point, the probability that this item belongs to $L_1$ is exactly $\frac{|L_1|}{|S|} \leq \frac{R^*}{|S|}$, where the inequality follows from the fact that $L_1$ starts with size $R^*$. 
In other words, we have $\E[X_{i', r'}] \leq \frac{R^*}{|S|}$ when $i'\ne 1$. 
Note also that for $r' \leq r_0$, we always have\footnote{Indeed, since there are $R^* - r_0 > R^*/3 + 1$ rounds remaining (not including the current round), the number of items left is at least $n \cdot R^*/3 \geq m^*/3$.} $|S| \geq m^*/3$. 
From this, we can derive
\begin{align*}
\E\left[\sum_{r' \in [r_0]} \sum_{i' \in N} X_{i', r'} \right] 
\leq r_0 + \E\left[\sum_{r' \in [r_0]} \sum_{i' \in N \setminus \{1\}} \frac{R^*}{m^*/3} \right] 
&\leq r_0\left(1 + \frac{3 n R^*}{m^*}\right) \\
&\leq r_0 \cdot 10 
\leq 0.5R^*,
\end{align*}
which concludes our proof. $\hfill \square$
\end{proof}

Similarly to \Cref{cor:noiseless-value-upper}, for value queries, we can query all $nm$ values and run the algorithm from \Cref{thm:noiseless-comparison-upper-random}.

\begin{corollary}
\label{cor:noiseless-value-upper-random}
For uniformly random preferences, under the noiseless value query model, there exists a deterministic algorithm that outputs the round-robin allocation using $O(nm)$ queries and expected $O(nm + m \log m)$ time.
\end{corollary}

\subsection{Lower Bounds}
\label{subsec:lb-noiseless}

We now turn to lower bounds.
First, we present a lower bound of $\Omega(nm)$ for comparison queries.

\begin{theorem}
\label{thm:noiseless-comparison-lower-nm}
Under the noiseless comparison query model, any (possibly randomized) algorithm that outputs the round-robin allocation with probability at least $2/3$ makes $\Omega(nm)$ queries in expectation.\footnote{All of our lower bounds in this section hold even when $2/3$ is replaced by any constant strictly larger than $1/2$.}
\end{theorem}

Before we proceed, we introduce some additional notation. 
Let $\Alg$ be an algorithm for the round-robin problem. 
A triplet $a = (i,j,j')$ represents a query in which $\Alg$ asks if agent~$i$ prefers item $j$ to $j'$.  
Let $\Ac$ be the set of all such triplets, with $|\Ac| = n{m \choose 2}$; since querying $(i,j_1,j_2)$ is equivalent to querying $(i,j_2,j_1)$ and flipping the answer, we omit such ``duplicate'' triplets from $\Ac$ without loss of generality.

Define an \emph{instance} $\nu$ of our problem to be a setting of the agent preferences, and let $L_{\nu}$ be the correct round-robin allocation 
when the instance is $\nu$. 
For any instance $\nu$ and agent $i$, let $\Jc_{i}(\nu)$ be the set of all items~$j$ such that in the correct round-robin procedure, item $j$ is not allocated to any of the agents $1,2,\dots,i-1$ in the first round and is not allocated to agent $i$ in any round.
We will use the following lemma.

\begin{lemma} \label{lem:count-alloc-items}
Let $\nu$ be an arbitrary instance of our round-robin problem. 
Then, $\sum_{i\in N} |\Jc_i(\nu)| \ge nm/4.$
\end{lemma}

\begin{proof}
We write $k_i := 1 + \lfloor (m - i) / n \rfloor$ to denote the number of items allocated to agent $i$ across all rounds. 
Since $i - 1$ items are allocated to agents $1, \dots, i - 1$ in the first round, and $k_i$ items are allocated to agent $i$ across all rounds, we have $|\Jc_i(\nu)| \geq m - (i - 1) - k_i$.
Summing this over all $i \in N$, we get
\begin{align*}
\sum_{i\in N} |\Jc_i(\nu)| 
&\ge nm - \frac{n(n - 1)}{2} - m = (n - 1)\left(m - \frac{n}{2}\right) \geq \frac{n}{2}\cdot\frac{m}{2} = \frac{nm}{4},
\end{align*}
as desired. $\hfill \square$
\end{proof}

For brevity, we say that an algorithm $\Alg$ is \emph{$\alpha$-correct} if $\Pr_{\Alg}[\Alg(\nu) = L_\nu] \geq \alpha$ for any instance $\nu$, where the probability is taken over the randomness of $\Alg$.
Moreover, for a distribution $\cD$ over instances, we say that $\Alg$ is \emph{$(\alpha, \cD)$-correct} if $\Pr_{\nu \sim \cD,\, \Alg}[\Alg(\nu) = L_\nu] \geq \alpha$, where the probability is taken over both the random instance $\nu$ drawn from~$\cD$ and the randomness of $\Alg$.
We will also use the following lemma, which is in the spirit of Yao's principle.

\begin{lemma} \label{lem:yao}
If there exists a $2/3$-correct algorithm using at most $q$ queries in expectation for some $q\in\mathbb{R}^+$, then for any distribution $\cD$ of instances, there exists a deterministic algorithm that makes $O(q)$ queries in the worst case and is $(0.99, \cD)$-correct.
\end{lemma}

\begin{proof}
Denote by $\Alg$ the $2/3$-correct algorithm that uses at most
$q$ queries in expectation. 
Then, let $\Alg'$ be the algorithm that runs $\Alg$ $10000$ times and outputs the result that appears most often (with ties broken arbitrarily). 
$\Alg'$ uses at most $q' := 10000q$ queries in expectation and, by standard concentration bounds, $\Alg'$ is $0.9999$-correct.\footnote{By a Chernoff bound, the actual value is at least $1 - e^{-200}$.}
Let $\Alg'_r$ denote the algorithm $\Alg'$ when we fix the randomness of the algorithm to be $r$, and let $\tau_r(\nu)$ denote the number of queries made by $\Alg'_r$ on input $\nu$. 
By the aforementioned guarantees, we have $\Pr_{r}[\Alg'_r(\nu) = L_\nu] \geq 0.9999$ and $\E_{r}[\tau_r(\nu)] \leq q'$ for all $\nu$.
Therefore, by Markov's inequality, we have
\begin{align*}
\Pr_r[\Pr_{\nu \sim \cD}[\Alg'_r(\nu) = L_\nu] \geq 0.999]
&= 1 - \Pr_r[\Pr_{\nu \sim \cD}[\Alg'_r(\nu) \ne L_\nu] > 0.001] \\
&\geq 1 - \frac{\E_r[\Pr_{\nu \sim \cD}[\Alg'_r(\nu) \ne L_\nu]]}{0.001} \\
&\geq 1 - \frac{0.0001}{0.001} = 0.9,
\end{align*}
and
\begin{align*}
\Pr_r[\E_{\nu \sim \cD}[\tau_r(\nu)] \leq 2q'] 
&= 1 - \Pr_r[\E_{\nu \sim \cD}[\tau_r(\nu)] > 2q'] \\
&\geq 1 - \frac{\E_{r,\, \nu \sim \cD}[\tau_r(\nu)]}{2q'} \\
&\geq 1 - \frac{q'}{2q'} = 0.5.
\end{align*}
Thus, there must exist $r^*$ such that $\Pr_{\nu \sim \cD}[\Alg'_{r^*}(\nu) = L_\nu] \geq 0.999$ and $\E_{\nu \sim \cD}[\tau_{r^*}(\nu)] \leq 2q'$.

Finally, let $\Alg''$ be the deterministic algorithm that runs $\Alg'_{r^*}$ except that it fails (or outputs an arbitrary allocation) whenever $\Alg'_{r^*}$ tries to make more than $2000q'$ queries. 
Clearly, $\Alg''$ uses at most $2000q' = O(q)$ queries in the worst case. 
Furthermore, 
\begin{align*}
\Pr_{\nu \sim \cD}[\Alg''(\nu) = L_{\nu}] 
&\geq \Pr_{\nu \sim \cD}[\Alg'_{r^*}(\nu) = L_{\nu}] - \Pr_{\nu \sim \cD}[\tau_{r^*}(\nu) > 2000q'] \\
&\geq 0.999 - 0.001 > 0.99,
\end{align*}
where the second inequality follows from our choice of $r^*$ and Markov's inequality.
This means that $\Alg''$ is $(0.99, \cD)$-correct. $\hfill \square$
\end{proof}

Given \Cref{lem:yao}, to prove \Cref{thm:noiseless-comparison-lower-nm}, the following lower bound against deterministic algorithms is sufficient.

\begin{proposition}
\label{prop:noiseless-comparison-lower-nm-det}
Under the noiseless comparison query model, there exists a distribution $\cD$ over instances such that any deterministic $(0.99, \cD)$-correct algorithm makes $\Omega(nm)$ queries in the worst case.
\end{proposition}

\begin{proof}
Let $\cD$ be the distribution based on uniformly random preferences.
Suppose for contradiction that there exists a deterministic algorithm $\Alg$ that is $(0.99, \cD)$-correct and makes at most $q := 0.01nm$ queries in the worst case.

For any instance $\nu$, agent $i \in N$, and item $j \in M$, let $q_{i, j}(\nu)$ be the indicator variable of whether the pair $(i,j)$ is involved in any query made by $\Alg$ when run on $\nu$.
Furthermore, let $\nu^{i,j}$ denote the instance that is the same as $\nu$ except that item~$j$ is made the most preferred item of agent~$i$.
To help with the proof of \Cref{prop:noiseless-comparison-lower-nm-det}, we will use the following lemma.

\begin{lemma} \label{lem:switch-to-max-incorrect}
For any algorithm $\Alg$, it holds that $q_{i, j}(\nu) = 1$ or $j \notin \Jc_i(\nu)$ or $\Alg(\nu) \ne L_\nu$ or $\Alg(\nu^{i, j}) \ne L_{\nu^{i, j}}$.
\end{lemma}

\begin{proof}
Suppose that $q_{i, j}(\nu) = 0$, i.e., the pair $(i,j)$ is not involved in a query by $\Alg(\nu)$. 
This implies that $\Alg(\nu) = \Alg(\nu^{i,j})$ because all comparison queries not involving $(i,j)$ result in the same answer in both $\nu$ and $\nu^{i, j}$. 
Observe that if $j \in \Jc_i(\nu)$, then $j$ is not allocated to $i$ in $L_{\nu}$ but is allocated to $i$ in $L_{\nu^{i, j}}$. 
Thus, $L_{\nu} \ne L_{\nu^{i, j}}$. 
It follows that at least one of $\Alg(\nu) \ne L_\nu$ and $\Alg(\nu^{i, j}) \ne L_{\nu^{i, j}}$ must hold.
$\hfill \square$
\end{proof}

Next, observe that if we pick $\nu \sim \cD$, $i \in N$, and $j \in M$ uniformly and independently at random, then $\nu^{i, j}$ has the same distribution as $\cD$, due to symmetry. 
Hence, picking $\nu, i, j$ in this way, we have
\begin{align*}
2\Pr_{\nu}[\Alg(\nu) \ne L_{\nu}] 
&= \Pr_{\nu}[\Alg(\nu) \ne L_{\nu}] + \Pr_{\nu, i, j}[\Alg(\nu^{i, j}) \ne L_{\nu^{i, j}}] \\
&= \E_{\nu, i, j}\left[\ind\left[\Alg(\nu) \ne L_{\nu}\right] + \ind\left[\Alg(\nu^{i, j}) \ne L_{\nu^{i, j}}\right]\right] \\
&\geq \E_{\nu, i, j}[1 - q_{i, j}(\nu) - \ind[j \notin \Jc_i(\nu)]] \\
&= \E_{\nu, i, j}\left[\ind[j \in \Jc_i(\nu)] - q_{i, j}(\nu)\right] \\
&\geq \frac{1}{nm}\E_{\nu}\left[\sum_{i \in N} |\Jc_i(\nu)|\right] - \frac{2q}{nm} 
\geq \frac{1}{4} - 0.02
> 0.2,
\end{align*}
where the first and third inequalities follow from \Cref{lem:switch-to-max-incorrect} and \Cref{lem:count-alloc-items}, respectively, and the factor of $2$ in the second inequality arises because each query $(i,j,j')$ can contribute to both $q_{i,j}(\nu)$ and $q_{i,j'}(\nu)$.
This contradicts our assumption that $\Alg$ is $(0.99, \cD)$-correct. $\hfill \square$
\end{proof}

The proof of an analogous bound for value queries is essentially the same.
Note that since $nm$ value queries are clearly sufficient, this bound cannot be improved.

\begin{corollary}
\label{cor:noiseless-value-lower-nm}
Under the noiseless value query model, any (possibly randomized) algorithm that outputs the round-robin allocation with probability at least $2/3$ makes $\Omega(nm)$ queries in expectation.
\end{corollary}

Next, we prove a bound of $\Omega(m\log m)$ for comparison queries---by \Cref{thm:noiseless-comparison-upper}, this bound is tight for constant~$n$.

\begin{theorem}
\label{thm:noiseless-comparison-lower-mlogm}
Under the noiseless comparison query model, any (possibly randomized) algorithm that outputs the round-robin allocation with probability at least $2/3$ makes $\Omega(m \log m)$ queries in expectation.
\end{theorem}

Given \Cref{lem:yao}, to prove \Cref{thm:noiseless-comparison-lower-mlogm}, it suffices to show the following bound against deterministic algorithms.

\begin{proposition}
\label{prop:noiseless-comparison-lower-mlogm-det}
Under the noiseless comparison query model, there is a distribution $\cD$ over identical-preference instances such that any deterministic $(0.99, \cD)$-correct algorithm makes $\Omega(m \log m)$ queries in the worst case.
\end{proposition}

Since the proof of \Cref{prop:noiseless-comparison-lower-mlogm-det} only uses identical preferences, we do not distinguish between queries for different agents and view each comparison simply as a tuple $(j, j')$ of items. 
Further, we represent an identical-preference instance $\nu$ by a permutation $\sigma: [m] \to [m]$, where item $j$ is preferred to $j'$ (by all agents) exactly when $\sigma(j) > \sigma(j')$. 
Let $L_\sigma$ be the correct round-robin allocation when the instance is $\sigma$.
Let $\cR$ denote the complete set of comparison results, i.e., $\cR = \{(j, j', r) \mid 1 \leq j < j' \leq m,\, r \in \{0, 1\}\}$, where $(j, j', 1)$ means that $j$ is preferred to $j'$ and $(j, j', 0)$ indicates the opposite preference. 
For any set $R \subseteq \cR$ of comparison query results, let $\Xc(R)$ be the set of all permutations on $[m]$ that are compatible with $R$. 
We write $\sigma \sim \Xc(R)$ to signify a permutation drawn uniformly at random from $\Xc(R)$. 
Notice that $|\Xc(R)| = 1$ if and only if the comparison results in $R$ completely determine the ordering of items. 
Our main lemma is that, unless this is the case, we cannot find an allocation that agrees with almost all permutations in $\Xc(R)$.

\begin{lemma} \label{lem:not-sort-incorrect}
For any $R \subseteq \cR$, if $|\Xc(R)| > 1$, then for any allocation $A$, we have $\Pr_{\sigma \sim \Xc(R)}[L_\sigma \ne A] \geq 3/44$.
\end{lemma}

The proof of \Cref{lem:not-sort-incorrect} involves showing that we can find an item $j \in M$ such that, for a random $\sigma \sim \Xc(R)$, the value $\sigma(j)$ is sufficiently random. 
We show this by leveraging results from the theory of posets and linear extensions~\cite{Kahn1984BalancingPE,Stanley81}. 
A sequence $a_1, \dots, a_m$ is called \emph{unimodal} if there exists $\ell \in [m]$ such that $a_1 \leq \dots \leq a_\ell$ and $a_\ell \geq a_{\ell + 1} \geq \cdots \geq a_m$. 
We will use the following fundamental result shown by Stanley~\cite{Stanley81}  via a deep connection to mixed volume in geometry.\footnote{Note that Stanley~\cite{Stanley81} and subsequent authors stated their results in terms of \emph{posets} and \emph{linear extensions}, but these correspond precisely to the set $R$ of comparison results and a permutation in $\Xc(R)$, respectively.}

\begin{theorem}[\cite{Stanley81}] \label{thm:unimodal}
For any $R \subseteq \cR$ and any item $j$, the sequence $a_1, \dots, a_m$ defined by $a_k = \Pr_{\sigma \sim \Xc(R)}[\sigma(j) = k]$ for $k\in[m]$ is unimodal.
\end{theorem}

Using (a stronger version of) \Cref{thm:unimodal},\footnote{Namely, Stanley~\cite{Stanley81} also showed that the sequence is \emph{log-concave}. 
We do not state this property as it is not required for our purposes.} Kahn and Saks~\cite{Kahn1984BalancingPE} proved the following result implying that as long as the item order is not yet completely determined by the comparison results so far, there exist two items $j_1, j_2$ such that if we query these two items, then we are guaranteed to reduce the size of $\Xc(R)$ by at least a constant factor. 
The constant factor obtained by Kahn and Saks~\cite{Kahn1984BalancingPE} was $8/11$; it was improved to roughly $0.724$ by Brightwell et al.~\cite{Brightwell1995}.\footnote{The famous \emph{$\frac{1}{3}$-$\frac{2}{3}$ conjecture} states that this constant can be improved to $2/3$, which would be tight~\cite{Kislitsyn1968AFP}.}

\begin{theorem}[\cite{Kahn1984BalancingPE}] \label{thm:balancing-pe}
For any $R \subseteq \cR$, if $|\Xc(R)| > 1$, then there exist $j_1, j_2 \in [m]$ such that
$3/11 \leq \Pr_{\sigma\sim\Xc(R)}[\sigma(j_1) > \sigma(j_2)] \leq 8/11$.
\end{theorem}

From these two theorems, we can derive the following lemmas.

\begin{lemma} \label{lem:anti-concen-single-coord}
For any $R \subseteq \cR$, if $|\Xc(R)| > 1$, then there exists $j \in [m]$ such that $\max_{k \in [m]} \Pr_{\sigma\sim\Xc(R)}[\sigma(j) = k] \leq 19/22$.
\end{lemma}

\begin{proof}
Suppose for contradiction that, for every $j \in [m]$, there exists $k_j \in [m]$ such that $\Pr_{\sigma\sim\Xc(R)}[\sigma(j) = k_j] > 19/22$. For any $j_1, j_2 \in [m]$, if $k_{j_1} > k_{j_2}$, then we have
\begin{align*}
\Pr_{\sigma\sim\Xc(R)}[\sigma(j_1) > \sigma(j_2)] 
&\geq \Pr_{\sigma\sim\Xc(R)}[\sigma(j_1) = k_{j_1} \, \wedge \, \sigma(j_2) = k_{j_2}] \\
&= 1 - \Pr_{\sigma\sim\Xc(R)}[\sigma(j_1) \ne k_{j_1} \, \vee \,  \sigma(j_2) \ne k_{j_2}] \\
&\geq 1 - \Pr_{\sigma\sim\Xc(R)}[\sigma(j_1) \ne k_{j_1}] - \Pr_{\sigma\sim\Xc(R)}[\sigma(j_2) \ne k_{j_2}] \\
&> 1 - \frac{3}{22} - \frac{3}{22} = \frac{8}{11}.
\end{align*}
Similarly, if $k_{j_1} \leq k_{j_2}$, then we have
\begin{align*}
\Pr_{\sigma\sim\Xc(R)}[\sigma(j_1) > \sigma(j_2)] 
&\leq 1 - \Pr_{\sigma\sim\Xc(R)}[\sigma(j_1) = k_{j_1} \, \wedge \,  \sigma(j_2) = k_{j_2}] \\
&= \Pr_{\sigma\sim\Xc(R)}[\sigma(j_1) \ne k_{j_1} \, \vee \,  \sigma(j_2) \ne k_{j_2}] \\
&\leq \Pr_{\sigma\sim\Xc(R)}[\sigma(j_1) \ne k_{j_1}] + \Pr_{\sigma\sim\Xc(R)}[\sigma(j_2) \ne k_{j_2}] \\
&< \frac{3}{22} + \frac{3}{22} = \frac{3}{11}.
\end{align*}
Since this holds for all $j_1, j_2 \in [m]$, \Cref{thm:balancing-pe} is violated and we arrive at a contradiction. $\hfill \square$
\end{proof}

\begin{lemma} \label{lem:mod-anti-concen}
For any $R \subseteq \cR$ with $|\Xc(R)| > 1$, there exists $j \in [m]$ such that $\max_{r \in [n]} \Pr_{\sigma \sim \Xc(R)}[\sigma(j) \equiv r \Mod{n}] \leq 41/44$.
\end{lemma}

\begin{proof}
By \Cref{lem:anti-concen-single-coord}, there exists $j \in [m]$ such that $\max_{k \in [m]} \Pr_{\sigma\sim\Xc(R)}[\sigma(j) = k] \leq 19/22$. 
Consider the sequence $a_1, \dots, a_m$ defined by $a_\ell := \Pr_{\sigma \sim \Xc(R)}[\sigma(j) = \ell]$ for $\ell\in [m]$; note that $\sum_{\ell\in[m]} a_\ell = 1$.
\Cref{thm:unimodal} implies that the sequence is unimodal; let $\ell^* := \argmax_{\ell \in [m]} a_\ell$ denote its peak. 
Note that we have $a_{\ell^*} \leq 19/22$.

Consider any $r \in [n]$.  Let $\kappa^* := \max\{0, \lfloor (\ell^* - r) / n \rfloor\}$, and let $\kappa := \lfloor (m - r) / n \rfloor$.
We have
\begin{align*}
&2\Pr_{\sigma \sim \Xc(R)}[\sigma(j) \equiv r \Mod{n}] \\
&= 2\sum_{q = 0}^{\kappa} a_{nq + r} \\
&= 2a_{n\kappa^* + r} + \sum_{q=0}^{\kappa^* - 1} 2a_{nq + r} + \sum_{q=\kappa^*+1}^{\kappa} 2a_{nq + r} \\
&\leq 2a_{n\kappa^* + r} + \sum_{q=0}^{\kappa^* - 1} (a_{nq + r} + a_{nq + r + 1}) + \sum_{q=\kappa^*+1}^{\kappa} (a_{nq + r - 1} + a_{nq + r}) \\
&\leq a_{n\kappa^* + r} + \sum_{\ell \in [m]} a_\ell \\
&\leq a_{\ell^*} + 1\\
&\leq \frac{19}{22} + 1,
\end{align*}
where the first inequality is due to the unimodality of the sequence and the third inequality is due to $\ell^*$ being the peak. 
Dividing both sides by $2$ yields the desired conclusion. $\hfill \square$
\end{proof}

We are now ready to prove \Cref{lem:not-sort-incorrect}.

\begin{proof}[of \Cref{lem:not-sort-incorrect}]
From \Cref{lem:mod-anti-concen}, there exists $j \in [m]$ such that $\Pr_{\sigma \sim \Xc(R)}[\sigma(j) \equiv r \Mod{n}] \leq 41/44$ for all $r \in [n]$. 
Let $i$ denote the agent to whom item~$j$ is assigned according to the allocation $A$. 
We have $\Pr_{\sigma \sim \Xc(R)}[L_\sigma \ne A] \geq \Pr_{\sigma \sim \Xc(R)}[\sigma(j) \not\equiv i \Mod{n}] \geq 3/44$. $\hfill \square$
\end{proof}

\Cref{lem:not-sort-incorrect} implies that we have to determine~$\sigma$ with sufficiently high probability using the comparison queries. 
Without the ``sufficiently high probability'' part, this is exactly the sorting problem, for which it is well-known 
that $\Omega(m \log m)$ queries are required. 
To establish \Cref{prop:noiseless-comparison-lower-mlogm-det}, we show that a similar number of queries is still necessary even with the ``high probability'' relaxation.

\begin{proof}[of \Cref{prop:noiseless-comparison-lower-mlogm-det}]
Let $\cD$ be the distribution based on identical preferences such that the preference order of items is uniformly random.
Suppose for contradiction that there exists a deterministic algorithm $\Alg$ that is $(0.99, \cD)$-correct and makes at most $q := 0.1 m \log_2 m$ queries in the worst case. 
We use the standard representation of $\Alg$ as a binary decision tree:\footnote{See, e.g., Section~8.1 of Cormen et al.~\cite{CLRS-book}.} each internal node of the tree corresponds to a comparison query $(j, j')$, and the left and right children correspond to the query answer being $0$ and $1$, respectively. 
Let $\Lambda$ denote the set of leaves of the tree; since $\Alg$ makes at most $q$ queries, $|\Lambda| \leq 2^q$. 
Each leaf $\lambda \in \Lambda$ corresponds to the algorithm's termination, at which point it outputs some allocation $A^\lambda$.
We use $R_\lambda \subseteq \cR$ to denote the set of comparison results leading to the leaf $\lambda$. 
Finally, let $\lambda(\sigma)$ denote the leaf that $\Alg$ ends up in when run on~$\sigma$.

We can now bound the probability that $\Alg$ is incorrect on a random $\sigma \sim \cD$ as follows:
\begin{align*}
\Pr_{\sigma \sim \cD}[L_{\sigma} \ne \Alg(\sigma)]
&= \Pr_{\sigma \sim \cD}[L_{\sigma} \ne A^{\lambda(\sigma)}] \\
&= \sum_{\lambda \in \Lambda} \Pr_{\sigma \sim \cD}[\lambda(\sigma) = \lambda] \cdot \Pr_{\sigma \sim \cD}[L_{\sigma} \ne A^{\lambda} \mid \lambda(\sigma) = \lambda] \\
&= \sum_{\lambda \in \Lambda} \frac{|\Xc(R_{\lambda})|}{m!} \cdot \Pr_{\sigma \sim \Xc(R_{\lambda})}[L_{\sigma} \ne A^{\lambda}] \\
&\geq \sum_{\lambda \in \Lambda} \frac{|\Xc(R_{\lambda})|}{m!} \cdot \frac{3}{44} \cdot \left(1 - \ind[|\Xc(R_\lambda)| = 1]\right) \\
&= \frac{3}{44} - \frac{3}{44} \sum_{\lambda \in \Lambda} \frac{\ind[|\Xc(R_\lambda)| = 1]}{m!} \\
&\geq \frac{3}{44}\left(1 - \frac{|\Lambda|}{m!}\right) \geq \frac{3}{44}\left(1 - \frac{2^q}{m!}\right) > 0.01,
\end{align*}
where the first inequality follows from \Cref{lem:not-sort-incorrect} and the last inequality from $q = 0.1 m \log_2 m$. 
This contradicts our assumption that $\Alg$ is $(0.99,\cD)$-correct. $\hfill \square$
\end{proof}

\section{Noisy Setting}
\label{sec:noisy}

In this section, we turn our attention to the noisy setting.
Because of the noise, we cannot expect algorithms to always output the correct answer.
Therefore, we will instead require them to be correct with probability at least $1-\delta$, for a given parameter~$\delta$.
Throughout the section, we adopt the mild assumption that $\delta \in (0, 1/2 - c)$ for some constant $c > 0$, and we treat the noise parameter $\rho$ as a fixed constant in $(0,1/2)$ (not scaling with $n$ and $m$).

\subsection{Upper Bounds}
\label{sec:noisy-upper}

We start with a simple upper bound for comparison queries.

\begin{theorem}
\label{thm:noisy-comparison-upper}
Under the noisy comparison query model, there exists a deterministic algorithm that outputs the round-robin allocation with probability at least $1-\delta$ using $O(nm\log(m/\delta))$ queries and $O(nm\log(m/\delta))$ time.
\end{theorem}

\begin{proof}
For each agent, sort the items according to her preferences using a noisy sorting algorithm; then, allocate the items using the $n$ resulting sorted lists.  
Noisy sorting is a well-studied problem, and for our proof, it suffices to use the algorithm of Feige et al.~\cite[Thm.~3.2]{FeigeRaPe94}, which requires $O(m\log(m/\delta_0))$ queries and time\footnote{Feige et al.~\cite{FeigeRaPe94} did not make an explicit claim on time. 
However, one can observe that the time complexity of their algorithm is the same as its query complexity.} to correctly sort $m$ items with probability at least $1-\delta_0$, for any $\delta_0 \in (0,1/2)$. 

Since we have $n$ lists to be sorted, we set $\delta_0 = \delta/n$, so that a union bound yields an overall success probability of $1-\delta$.  Hence, the overall complexity is $O(nm\log(nm/\delta))$, which is equivalent to $O(nm\log(m/\delta))$ due to the fact that $m/\delta \le nm/\delta \le m^2/\delta^2$ (recall our assumption $n \le m$).    $\hfill \square$ 
\end{proof}

For value queries, we query each agent's utility for each item a sufficient number of times and run our noiseless algorithm based on the majority values.

\begin{theorem}
\label{thm:noisy-value-upper}
Under the noisy value query model, there exists a deterministic algorithm that outputs the round-robin allocation with probability at least $1-\delta$ using $O(nm\log(m/\delta))$ queries and $O(nm\log(m/\delta))$ time.
\end{theorem}

\begin{proof}
For each agent $i \in N$ and item $j \in M$, we query $u_i(j)$ for a fixed number $T$ of times, and declare its value to be the majority value observed, breaking ties arbitrarily.
We then run our noiseless round-robin algorithm (\Cref{cor:noiseless-value-upper}) assuming these declared values to be the correct utilities. 
It remains to determine a choice of $T$ such that all declared utilities are correct with (joint) probability at least $1-\delta$.  
For this correctness, it suffices that fewer than $T/2$ of the $T$ queries to each pair are corrupted by the adversary.

For each query, the probability of the observed value being corrupted is $\rho \in (0, 1/2)$. 
Hence, by Hoeffding's inequality, after $T$ queries, the probability of having $T/2$ or more corrupted values is at most $e^{-2T(\frac{1}{2} - \rho)^2}$, which can be upper-bounded by any specified $\delta_0 \in (0,1)$ if we choose $T= \left\lceil \frac{ \log(1/\delta_0) }{2(1/2 - \rho)^2} \right\rceil \in \Theta(\log(1/\delta_0))$.  
By a union bound over the $nm$ pairs $(i,j)$, the overall error probability is at most $\delta$ provided that $\delta_0 = \delta/(nm)$, which gives $T = \Theta(\log(nm/\delta))$.  
Since $n \le m$, the overall query complexity is therefore $nmT = O(nm\log(m/\delta))$.  
Moreover, one can easily check that the time complexity is of the same order as the query complexity. $\hfill \square$
\end{proof}

While the proofs of our upper bounds are simple, based on our current understanding, it is conceivable that these bounds are already optimal.
In \Cref{app:challenges}, we discuss some challenges that we faced when trying to improve the bounds.

\subsection{Lower Bounds}

Next, we shift our focus to lower bounds.
We derive a bound of $\Omega(nm\log(1/\delta))$ for comparison queries.

\begin{theorem}
\label{thm:noisy-comparison-lower-nm}
Under the noisy comparison query model, any (possibly randomized) algorithm that outputs the round-robin allocation with probability at least $1-\delta$ makes $\Omega(nm\log(1/\delta))$ queries in expectation.
\end{theorem}

\begin{proof}
We use the same notation as in \Cref{subsec:lb-noiseless}. 
We interpret the query model as a multi-armed bandit problem \cite{LattimoreSz20} (with a highly unconventional objective) in which each $a \in \Ac$ is an ``arm'' or ``action''.  
If the $t$-th query made is $a_t = (i,j,j')$, then the resulting observation is denoted by $y_t$, and is drawn from the distribution ${\rm Bernoulli}(1-\rho)$ if $i$ prefers $j$ to $j'$, and ${\rm Bernoulli}(\rho)$ otherwise.  
Let $P_a$ denote the (Bernoulli) distribution associated with action $a$, i.e., it holds for any query index $t$ and $y \in \{0,1\}$ that $P_a(y) = \PP[y_t = y \mid a_t = a]$.

Let $\PP^{\nu}$ and $\EE^{\nu}$ denote the probability and expectation (with respect to the randomness in the algorithm and/or the query answers), respectively, when the underlying instance is $\nu$.
By assumption, we have $\PP^{\nu}[ \mathsf{Alg}(\nu) = L_{\nu}] \ge 1-\delta$ for all $\nu$.  
Note that the number of queries taken when $\mathsf{Alg}$ terminates is a random variable, as this may depend on the observed $y_t$ values (which are themselves random) and moreover $\mathsf{Alg}$ itself may be randomized.

With the actions $\{a_t\}$ and query responses $\{y_t\}$ being interpreted under a bandit framework as above, we can make use of a highly general result from the bandit literature that translates into an information-theoretic lower bound on the number of times certain actions are played (in expectation).  
These bounds are expressed in terms of the \emph{KL divergence} $D(P\|Q) = \sum_{x} P(x)\log\frac{P(x)}{Q(x)}$ for probability mass functions $P,Q$, and its binary version $d(a,b) = a\log\frac{a}{b} + (1-a)\log\frac{1-a}{1-b}$ for real numbers $a,b \in [0,1]$ (i.e., $d(a,b)$ is the KL divergence between ${\rm Bernoulli}(a)$ and ${\rm Bernoulli}(b)$ distributions). 
This result and similar variations have been used in numerous bandit works.\footnote{For example, see Lemma~1 of Kaufmann et al.~\cite{KaufmannCaGa16} and Exercise~15.7 of Lattimore and Szepesv\'ari~\cite{LattimoreSz20}.}

\begin{lemma}
\label{lem:bandit-inequality}
Let $\nu$ and $\nu'$ be any two bandit instances defined on the same (finite) set of arms $\Ac$, with corresponding observation distributions $\{P_a\}_{a \in \Ac}$ and $\{P'_a\}_{a \in \Ac}$.
Let $\tau$ be the (random) total number of queries made when the algorithm terminates, and let $\Ec$ be any probabilistic event that can be deduced from the resulting history $(a_1,y_1,\dotsc,a_\tau,y_\tau)$, possibly with additional randomness independent of that history. 
Then, we have
	\begin{equation}
        \sum_{a \in \Ac} \EE^{\nu}[ T_a ] D(P_a \| P'_a) \ge d( \PP^{\nu}[\Ec], \PP^{\nu'}[\Ec] ), \label{eq:main_lemma}
    \end{equation} 
where $T_a$ is the (random) number of times action $a$ is queried up to the termination index $\tau$.
\end{lemma}

Intuitively, the right-hand side of \eqref{eq:main_lemma} identifies an event~$\Ec$ that (ideally) occurs with significantly different probabilities under the two instances (e.g., the algorithm outputting $L_{\nu}$ when $L_{\nu'} \ne L_{\nu}$).  
The left-hand side indicates that in order to permit such a difference in probabilities, actions with sufficient distinguishing power (i.e., high $D(P_a \| P'_a)$) must be played sufficiently many times (i.e., high $\EE^{\nu}[ T_a ]$).

Let $\nu$ be any instance of our round-robin problem. 
Recall \Cref{lem:count-alloc-items} (and the notation $\Jc_{i}(\nu)$), which asserts that
\begin{equation}
    \sum_{i\in N} |\Jc_i(\nu)| \ge \frac{nm}{4}. \label{eq:count}
\end{equation}

Now, for fixed $i$ and $j \in \Jc_i(\nu)$, consider a different instance $\nu'$ in which $j$ is made the most preferred item for agent $i$, and all other preferences remain unchanged.  
This means that $j$ is allocated to $i$ in $L_{\nu'}$, 
in particular implying that $L_{\nu} \ne L_{\nu'}$.  Moreover, in \eqref{eq:main_lemma}, we observe the following:
\begin{itemize}
    \item Unless the action $a$ corresponds to agent~$i$ and item~$j$ (along with some other arbitrary item), the quantity $D(P_a \| P'_a)$ is zero; this is due to our construction of $\nu'$ and the fact that $D(P\|P) = 0$ for any $P$.
    \item For \emph{any} action $a$, the quantity $D(P_a \| P'_a)$ is either zero or $d(\rho,1-\rho)$ (which is equal to $d(1-\rho,\rho)$), since our observation distributions are always ${\rm Bernoulli}(\rho)$ or ${\rm Bernoulli}(1-\rho)$. 
    \item Let $\Ec$ be the event that $\mathsf{Alg}$ outputs $L_{\nu'}$.
    The success condition on $\mathsf{Alg}$ implies that $\PP^{\nu}[\Ec] \le \delta$ and $\PP^{\nu'}[\Ec] \ge 1-\delta$.  
    Since $\delta < 1/2$, this in turn implies $d( \PP^{\nu}[\Ec], \PP^{\nu'}[\Ec] ) \ge d(\delta,1-\delta)$ by a standard monotonicity property of $d(a,b)$ \cite{KaufmannCaGa16}.
\end{itemize}
Combining the above findings, and defining $\Actil_{ij}$ to be the set of all actions involving agent $i$ and item $j$, we obtain $\sum_{a \in \Actil_{ij}} \EE^{\nu}[ T_a ] d(\rho,1-\rho) \ge d(\delta,1-\delta)$, or equivalently, $\sum_{a \in \Actil_{ij}} \EE^{\nu}[ T_a ]  \ge \frac{d(\delta,1-\delta)}{d(\rho,1-\rho)}$.
Since this holds for all pairs $(i,j)$ such that $j \in \Jc_i(\nu)$, we can sum over all such pairs and apply \eqref{eq:count} to obtain
\begin{equation}
    \sum_{i\in N} \sum_{j \in \Jc_i(\nu)} \sum_{a \in \Actil_{ij}} \EE^{\nu}[ T_a ] \ge \frac{nm}{4} \cdot \frac{d(\delta,1-\delta)}{d(\rho,1-\rho)}.  \label{eq:triple_sum}
\end{equation}

Next, we claim that the left-hand side of \eqref{eq:triple_sum} is upper-bounded by $2\EE^{\nu}[\tau]$, where $\tau$ is the (random) total number of queries.  
To see this, we upper-bound the summation $\sum_{j \in \Jc_i(\nu)}$ by $\sum_{j \in M}$, apply linearity of expectation, and observe that $\sum_{i\in N} \sum_{j\in M} \sum_{a \in \Actil_{ij}} T_a$ is exactly $2\tau$; the factor of $2$ arises because each query $(i,j_1,j_2)$ is counted twice (once when $j=j_1$ and once when $j=j_2$).  
It follows that 
\begin{equation*}
    \EE_{\nu}[\tau] \ge \frac{nm}{8}\cdot \frac{d(\delta,1-\delta)}{d(\rho,1-\rho)}. 
\end{equation*}
The proof is completed by recalling that $\rho$ is a fixed constant in $(0, 1/2)$, and noting that $d(\delta,1-\delta) \in \Omega(\log(1/\delta))$ since $\delta \le 1/2-c$. $\hfill \square$
\end{proof}

Next, we establish an analogous result for value queries.

\begin{theorem}
\label{thm:noisy-value-lower-nm}
Under the noisy value query model, any (possibly randomized) algorithm that outputs the round-robin allocation with probability at least $1-\delta$ makes $\Omega(nm\log(1/\delta))$ queries in expectation.
\end{theorem}

\begin{proof}
The majority of the steps are similar to the case of noisy comparisons (\Cref{thm:noisy-comparison-lower-nm}), so we focus on the differences and avoid repeating similar steps.

The actions are now of the form $a = (i,j)$, and with $\Ac$ again denoting the set of all actions, we have $|\Ac| = nm$. 
We consider a ``default'' instance $\nu$ in which each agent assigns utilities $1,2,\dots,m$ to the $m$ items, with utility $m$ for the most preferred, $m-1$ for the second-most preferred, and so on.  
We then consider a modified instance $\nu'$ 
with a single item $j \in \Jc_i(\nu)$ for some agent $i$ being made the most preferred, where $\Jc_i(\nu)$ is defined as in \Cref{lem:count-alloc-items}.
Specifically, in $\nu'$ we let $i$ have utility $m+1$ for this item $j$.

In the instances constructed above (which are the only ones we will use in this proof), each agent's set of utilities for the $m$ items is either $[m]$ (in $\nu$) or $[m+1] \setminus \{k\}$ for some $k \in [m]$ (in $\nu'$).  
Accordingly, the following adversary is well-defined:
\begin{itemize}
    \item If a query is made such that the true utility is different from $m+1$, then the adversary makes the query return $m+1$ (whenever given the chance).
    \item If a query is made such that the true utility is $m+1$, then the adversary makes the query return the unique value in $[m]$ that is not present in the agent's set of $m$ utilities (whenever given the chance).
\end{itemize}
Under this choice of adversary, letting $u^*$ be the value of $u_i(j)$ in instance $\nu$, we observe the following:
\begin{itemize}
    \item Under $\nu$, any query to $(i,j)$ outputs $u^*$ with probability $1-\rho$, and $m+1$ otherwise;
    \item Under $\nu'$, any query to $(i,j)$ outputs $m+1$ with probability $1-\rho$, and $u^*$ otherwise;
    \item If the query is not to $(i,j)$, then the resulting distributions are identical under $\nu$ and $\nu'$.
\end{itemize}
Since KL divergence only depends on the underlying probabilities and not the corresponding values (e.g., $\{0,1\}$ for Bernoulli vs.~$\{u^*,m+1\}$ described above), we find that in the first two bullet points, the action $a=(i,j)$ has corresponding distributions $(P_a,P'_a)$ in $(\nu,\nu')$ such that $D(P_a \| P'_a) = d(\rho,1-\rho) = d(1-\rho,\rho)$.  
Moreover, the third bullet point reveals that all other actions result in a KL divergence of zero.

With these observations in place, the same reasoning as in the proof of \Cref{thm:noisy-comparison-lower-nm} (including \Cref{lem:bandit-inequality}) yields the following: 
\begin{equation*}
    \EE^{\nu}[ T_a ]  \ge \frac{d(\delta,1-\delta)}{d(\rho,1-\rho)},
\end{equation*}
where $T_a$ denotes the number of times that $a = (i,j)$ is queried.  Summing over all $i$ and $j \in \Jc_i(\nu)$ and again using \eqref{eq:count}, we readily obtain $$\EE^{\nu}[\tau] \ge \frac{nm}{4}\cdot \frac{d(\delta,1-\delta)}{d(\rho,1-\rho)},$$ and the proof is completed by noting that $d(\delta,1-\delta) \in \Omega(\log(1/\delta))$. $\hfill \square$
\end{proof}

Finally, we derive a lower bound of $\Omega(m\log m)$ for both query models---the bound holds even for two agents with identical preferences.

\begin{theorem}
\label{thm:noisy-comparison-lower-mlogm}
Under the noisy comparison query model, for two agents, any (possibly randomized) algorithm that outputs the round-robin allocation with probability at least $1-\delta$ makes $\Omega(m\log(m/\delta))$ queries in expectation.
\end{theorem}

\begin{proof}
Consider two agents with identical preferences.
To simplify notation, we assume that $m$ is an even number; otherwise, the same proof applies with one particular item always being the least preferred.  

Let the ``default'' preference instance be the one in which $1 \succ 2 \succ \dotsc \succ m$, where $a \succ b$ means that item $a$ is preferred to item $b$.  
Then, consider alternative instances such that each pair of the form $(2k-1,2k)$ may be reversed, for $k\in [m/2]$.  
By taking all possible combinations of pairs that may be reversed, this creates $2^{m/2}$ total instances.

Consider a randomized setting in which the unknown instance is uniformly random over these $2^{m/2}$ instances, which is equivalent to each pair $(2k-1,2k)$ being reversed independently with probability $1/2$.  
For $k\in [m/2]$, let $S_k \sim \mathrm{Bernoulli}(1/2)$ be the indicator random variable of whether the corresponding reversal was done.  
See \Cref{tab:example} for  examples in the case $m=6$.

\begin{table}[!h]
    \centering
    \begin{tabular}{c|c}
        $[S_1,S_2,S_3]$& Instance \\
        \hline
        $[1,0,0]$ & $2\succ1\succ3\succ4\succ5\succ6$  \\
        $[1,1,0]$ & $2\succ1\succ4\succ3\succ5\succ6$ \\
        $[1,1,1]$ & $2\succ1\succ4\succ3\succ6\succ5$ \\
    \end{tabular}
    \caption{Examples of instances with $m=6$ in the proof of \Cref{thm:noisy-comparison-lower-mlogm}
    }
    \label{tab:example}
\end{table}

With the problem now being restricted to one of the $2^{m/2}$ instances, we observe that any algorithm for this problem can exhibit certain behavior without loss of optimality.
\begin{itemize}
    \item The algorithm only queries item pairs of the form $(2k-1,2k)$.  
    This is because the preference order between any other pair is identical under all the instances above (e.g., item $1$ is always preferred to item $3$), so the resulting answer conveys no useful information to the algorithm.
    \item For each pair $(2k-1,2k)$, the algorithm always assigns one of the items to agent~$1$ and the other to agent~$2$.  
    This is because the above construction ensures that the true round-robin allocation exhibits such behavior, so failing to follow it would guarantee being wrong.
\end{itemize}
In light of this, the allocation problem is now identical to the estimation of $S_1,\dotsc,S_{m/2}$: 
For each pair $(2k-1,2k)$, allocating item $2k-1$ to agent~$1$ (resp., agent~$2$) is equivalent to estimating $S_1$ to be $0$ (resp., $1$).  
Henceforth, we will use phrasing such as ``the algorithm's estimate of $S_k$'' with this equivalence in mind.

Let $\Hc_t$ denote the history of queries and answers up to (and including) the $t$-th query, and let $\Hc_t^{(k)}$ denote the subset of that history for which the queries are to the pair $(2k-1,2k)$.  
The following lemma reveals that under the randomized instance we are considering, the posterior distribution of $S_1,\dotsc,S_{m/2}$ (i.e., the conditional distribution given $\Hc_t$) exhibits a useful independence property.

\begin{lemma}
\label{lem:independence}
Under the uniform prior on $(S_1, \dots, S_{m/2})$, conditioned on any collection of query answers $\Hc_t$ up to some index~$t$, we
have that $(S_1, \dots, S_{m/2})$ are conditionally independent, and that the conditional distribution of each $S_k$ is the same as that conditioned on $\Hc_t^{(k)}$ alone.
\end{lemma}

\begin{proof}
    Equivalent statements can be found in existing works studying different problems,\footnote{For example, see Lemma~5 of Cai et al.~\cite{CaiLaSc23}.} but we provide a short proof for completeness.  
    We use an induction argument, with the base case being that before the first query is made, the variables $S_k$ are independent by assumption (and $\Hc_t = \Hc_0$ is empty).  
    For the induction step, suppose that the lemma is true for $t=t'-1$, where $t' \ge 1$.  Then, when $(2k_{t'}-1,2k_{t'})$ is queried for some $k_{t'}$, its answer trivially only depends on $S_{k_{t'}}$, and not on $\{S_k\}_{k \ne k_{t'}}$.  
    Hence, the posterior distribution of $S_k$ is changed given the new answer for $k = k_{t'}$, but stays unchanged for $k \ne k_{t'}$, which implies that the lemma holds for $t = t'$. $\hfill \square$
\end{proof}

Our goal is to lower-bound the expected number of queries for an algorithm that correctly estimates $(S_1,\dotsc,S_{m/2})$ with probability at least $1-\delta$.  
Let $\tau$ denote the (random) total number of queries, and let $\tau_k$ be the number of queries made to the pair $(2k-1,2k)$ (so that $\sum_{k\in[m/2]} \tau_k = \tau$).  
For convenience, we define $\Delta:=\EE[\tau]/m$.

The following straightforward lemma helps us move from average-case quantities (e.g., $\EE[\tau]$) to quantities that are bounded with a constant probability.

\begin{lemma}
\label{lem:tau}
Letting $\Delta=\frac{\EE[\tau]}{m}$, we have $\PP[\tau\le2m\Delta] \ge 1-\frac{\EE[\tau]}{2m\Delta} = \frac{1}{2}$. 
Furthermore:
\begin{itemize}
    \item[(i)] When $\tau\le 2m\Delta$, there exists a collection $\Kc$ of at least $m/4$ item pairs such that $\tau_k\le 8\Delta$ for each $k\in\Kc$;
    \item[(ii)] For the algorithm to attain an unconditional success probability of at least $1-\delta$, it must attain a conditional success probability of at least $1-2\delta$ conditioned on $\tau\le 2m\Delta$.
\end{itemize} 
\end{lemma}
\begin{proof}
    The claim on $\PP[\tau\le2m\Delta]$ is a direct application of Markov's inequality.  
    The subsequent two statements can be seen immediately via a proof by contradiction, where for statement (ii) we use the claim on $\PP[\tau\le2m\Delta]$. $\hfill \square$
\end{proof}

Next, we show that if the algorithm terminates with a certain number $\tau_k$ of queries made to the $k$-th pair, then the posterior uncertainty of $S_k$ satisfies a lower bound depending on $\tau_k$.

\begin{lemma}
\label{lem:error_k}
Let $\Hc_{\tau}$ be an arbitrary query history up to the termination time of the algorithm, and let $\tau_k$ be the number of queries to the pair $(2k-1,2k)$ in that history.  Then, for any estimate $\widehat{S}_k$ formed based on $\Hc_{\tau}$, we have
\begin{equation}
    \PP[ \widehat{S}_k \ne S_k \mid \Hc_{\tau} ] \ge \frac{1}{2} \rho^{\tau_k}, \nonumber
\end{equation}
where $\rho \in (0,1/2)$ is the noise level.
\end{lemma}

\begin{proof}
Let $\ell_k$ denote the number of $1$s observed for queries to $(2k-1,2k)$, meaning that the number of 0s for such queries is $\tau_k - \ell_k$.  
Recall from Lemma \ref{lem:independence} that $S_k$ has the same distribution given $\Hc_{\tau}$ as given $\Hc_{\tau}^{(k)}$.  
Then, since the prior distribution on $S_k$ is ${\rm Bernoulli}(1/2)$, the posterior distribution is as follows:
\begin{align}    
&\PP[S_k=0 \mid \Hc_{\tau}^{(k)}] \nonumber \\
&=\frac{\PP[S_k=0]\cdot\PP[\Hc_{\tau}^{(k)} \mid S_k=0]}{\PP[S_k=0]\cdot\PP[\Hc_{\tau}^{(k)} \mid S_k=0]+\PP[S_k=1]\cdot\PP[\Hc_{\tau}^{(k)} \mid S_k=1]} \label{eq:posterior1} \\
&=\frac{\frac{1}{2}\cdot\rho^{\ell_k}(1-\rho)^{\tau_k-\ell_k}}{\frac{1}{2}\cdot\rho^{\ell_k}(1-\rho)^{\tau_k-\ell_k}+\frac{1}{2}\cdot(1-\rho)^{\ell_k}\rho^{\tau_k-\ell_k}} \nonumber \\
&= \frac{1}{1+(\frac{1-\rho}{\rho})^{\ell_k}(\frac{\rho}{1-\rho})^{\tau_k-\ell_k}} \nonumber \\
&\ge \frac{1}{1+(\frac{1-\rho}{\rho})^{\tau_k}} \label{eq:posterior4} \\
&=\frac{\rho^{\tau_k}}{\rho^{\tau_k}+(1-\rho)^{\tau_k}} \nonumber \\
&\ge\frac{1}{2}\rho^{\tau_k}, \nonumber
\end{align}
where \eqref{eq:posterior1} follows from Bayes' rule and \eqref{eq:posterior4} holds because $\frac{1-\rho}{\rho} > 1$, which implies that the quantity $(\frac{1-\rho}{\rho})^{\ell_k}(\frac{\rho}{1-\rho})^{\tau_k-\ell_k} = (\frac{1-\rho}{\rho})^{2\ell_k}(\frac{\rho}{1-\rho})^{\tau_k}$ is highest when $\ell_k$ equals its highest possible value of $\tau_k$.

By the same reasoning, we also have $\PP[S_k=1\mid\Hc_{\tau}^{(k)}]\ge\frac{1}{2}\rho^{\tau_k}$.  
With the conditional probabilities of $S_k=0$ and $S_k=1$ satisfying the same lower bound $\frac{1}{2}\rho^{\tau_k}$, it follows that no matter what estimate is formed, it is incorrect with probability at least $\frac{1}{2}\rho^{\tau_k}$. $\hfill \square$
\end{proof}

Having established a lower bound on the conditional error probability for a single $k$, we now consider the \emph{joint} behavior of the $m/2$ values of $k$.  
To do so, we condition on $\tau \le 2m\Delta$.
By \Cref{lem:tau}(i), there exists a set $\Kc$ (which depends on $\Hc_{\tau}$) of size at least $m/4$ such that $\tau_k \le 8\Delta$ for each $k\in \Kc$. 
By \Cref{lem:tau}(ii), a conditional error probability exceeding $2\delta$ suffices to establish that the unconditional error probability exceeds $\delta$.
Letting $\Ec$ be the overall error event, and $\Ec_k$ be the event that $S_k$ is estimated incorrectly, we have
\begin{align}
    \PP[\Ec \mid \Hc_{\tau}]
        = \PP\left[ \bigcup_{k=1}^{m/2} \Ec_k \,\middle|\, \Hc_{\tau} \right] 
        &\ge \PP\left[ \bigcup_{k \in \Kc} \Ec_k \,\middle|\, \Hc_{\tau} \right] \nonumber \\
        &= 1 - \prod_{k \in \Kc}\Big(1 - \PP\big[ \Ec_k \mid \Hc_{\tau} \big]\Big) \label{eq:Ec_step3} \\
        &\ge 1 - \Big(1 - \frac{1}{2}\rho^{8\Delta}\Big)^{m/4} \label{eq:Ec_step4} \\
        &\ge 1 - \exp\bigg( - \frac{m}{8} \rho^{8\Delta} \bigg), \label{eq:Ec_step5} 
\end{align}
where \eqref{eq:Ec_step3} follows from the independence stated in Lemma \ref{lem:independence}, \eqref{eq:Ec_step4} follows from the lower bound in Lemma \ref{lem:error_k} along with $\tau_k \le 8\Delta$ and $|\Kc| \ge m/4$, and \eqref{eq:Ec_step5} follows since $1-z \le e^{-z}$ for all $z\in\mathbb{R}$.

Next, we show that when $\Delta$ is below a certain threshold, the right-hand side of \eqref{eq:Ec_step5} is greater than $2\delta$.  Specifically, we have
\begin{align*}
    1 - \exp\bigg( - \frac{m}{8} \rho^{8\Delta} \bigg) > 2\delta
    & \quad \iff \frac{m}{8} \rho^{8\Delta} > \log\frac{1}{1-2\delta} \\
    & \quad \iff 8\Delta \log \rho > \log\frac{8 \log \frac{1}{1-2\delta}}{m} \\
    & \quad \iff \Delta < \frac{ \log\frac{m}{8\log\frac{1}{1-2\delta}} }{8\log\frac{1}{\rho}}.
\end{align*}
Thus, we have shown the implication
\begin{equation}
    \Delta < \frac{ \log\frac{m}{8\log\frac{1}{1-2\delta}} }{8\log\frac{1}{\rho}} \implies \PP[\Ec \mid \Hc_{\tau}] > 2\delta. \label{eq:Delta_cond}
\end{equation}
Since this holds for \emph{any} $\Hc_{\tau}$ under the condition $\tau \le 2m\Delta$,
we deduce that attaining $\PP[\Ec] \le \delta$ requires the inequality on $\Delta$ in \eqref{eq:Delta_cond} to be reversed.  
As $\Delta = \EE[\tau]/m$, this implies that
\begin{equation}
    \EE[\tau] \ge \frac{ m \log\frac{m}{8\log\frac{1}{1-2\delta}} }{8\log\frac{1}{\rho}}. \nonumber
\end{equation}
Since we treat $\rho \in (0,1/2)$ as fixed, the desired $\Omega(m \log(m/\delta))$ lower bound is obtained by simply observing that 
(i)~$\log\frac{1}{1-2\delta} = \Theta(1)$ when $\delta$ is a fixed constant in $(0,1/2)$, and 
(ii) $\log\frac{1}{1-2\delta} = \Theta(\delta)$ when $\delta = o(1)$, which follows from a standard Taylor expansion. $\hfill \square$
\end{proof}

\begin{theorem}
\label{thm:noisy-value-lower-mlogm}
Under the noisy value query model, for two agents, any (possibly randomized) algorithm that outputs the round-robin allocation with probability at least $1-\delta$ makes $\Omega(m\log(m/\delta))$ queries in expectation.
\end{theorem}

\begin{proof}
The proof of \Cref{thm:noisy-comparison-lower-mlogm} applies to the setting of noisy value queries with minimal change.  
To see this, we assign values $\{m,m-1\}$ to the first pair of items, $\{m-2,m-3\}$ to the second pair of items, and so on, so that there are again $2^{m/2}$ potential instances.
We consider an adversary that always modifies a given value to the other of the two values that comprise a given pair (when given the chance).  
Hence, whenever a query is made to a given item in the pair, it observes the correct value of the item with probability $1-\rho$, and the incorrect value with probability $\rho$.  
This is mathematically equivalent to observing a Bernoulli random variable (with parameter $\rho$ or $1-\rho$) giving a noisy indication of which of the two items has the larger value, so the proof of \Cref{thm:noisy-comparison-lower-mlogm} is still valid. $\hfill \square$
\end{proof}

\section{Conclusion and Future Directions}

In this paper, we have analyzed the round-robin algorithm, one of the most widespread algorithms in the fair division literature, and presented several bounds on its complexity in the potential presence of noise.
Besides tightening the bounds themselves, our work opens up a number of appealing conceptual directions.
First, it would be interesting to explore the complexity of other fair division algorithms that rely only on ordinal rankings of items \cite{AzizGaMa15,BramsKiKl14}; 
such algorithms are cognitively less demanding for agents than algorithms that require either cardinal utilities or ordinal information on sets of items.
In addition, one could consider questions on fulfilling certain fairness notions in the presence of noise.
Another intriguing avenue is to consider other noise models, for example, a comparison query model in which the probability of error depends on how differently the relevant agent ranks the two queried items \cite{DavidsonKhMi14}.

\bibliographystyle{splncs04}
\bibliography{main}

\appendix

\section{Challenges in Obtaining Improved Noisy Upper Bounds}
\label{app:challenges}

While the proofs of our noisy upper bounds (\Cref{sec:noisy-upper}) are simple and may appear to leave room for improvements, we found that attaining such improvements is surprisingly challenging.  
In fact, based on our current understanding, it is conceivable that these simple upper bounds are already optimal.
In this appendix, we outline some difficulties that we encountered when trying to improve our noisy upper bounds via ideas similar to those used in our noiseless algorithms (\Cref{sec:noiseless-upper}).  
We focus on comparison queries, since they have received more attention in the related literature on algorithms with noise (whose results we would like to make use of).  

\subsection*{Worst-Case Preferences}

For worst-case preferences, our noiseless algorithm (\Cref{thm:noiseless-comparison-upper}) uses an $(m,n')$-quantiles subroutine that partitions the items into $m/n'$ disjoint groups such that each group is uniformly better than the next group, with complexity $O(m \log(m/n'))$; see \Cref{lem:partial-sort} (here we allow for $n'$ that may differ from $n$, though we set $n' = n$ in \Cref{lem:partial-sort}).  
Once this is done, it only remains to perform $m$ maximum-finding operations on lists of size up to $n'$.

An immediate difficulty is that we are not aware of any existing result for the quantiles subroutine in the noisy setting.\footnote{We could use (noisy) selection to solve the quantiles subproblem in a similar manner to the noiseless setting, but the resulting bound would be rather loose, and would certainly not lead to an improvement over Theorem \ref{thm:noisy-comparison-upper}.} 
Nevertheless, it is natural to expect a complexity of $O\big(m \log\frac{m/n'}{\delta_0} \big)$ queries for success probability $1-\delta_0$, analogous to the known results on finding the top $k$ items and their order using $O\big(m \log \frac{k}{\delta_0} \big)$ noisy queries (see Section 2.3 in the work of Cohen-Addad et al.~\cite{CohenaddadMaMa20}).  
We proceed with this natural conjecture despite the lack of a formal result.   
If we apply a union bound over $n$ agents, then we should set $\delta_0 = \delta/n$ (or smaller) to obtain an overall success probability of $1-\delta$, yielding $O\big(mn \log\frac{mn}{n' \delta} \big)$ queries in total.  
Thus, we need $n' \gg n$ to avoid being stuck with $O\big(mn \log\frac{m}{\delta} \big)$ scaling.

However, the noisy maximum operation on $n'$ items has complexity $O\big(n' \log \frac{1}{\delta_0} \big)$ for success probability $1-\delta_0$ \cite{FeigeRaPe94}, and with $m$ such operations, we should set $\delta_0 = \delta/m$ (or smaller) if a union bound is used.  
This leads to $O\big(mn' \log\frac{m}{\delta} \big)$ queries in total, which means we now need $n' \ll n$.  
Obviously, we cannot have $n' \ll n$ and $n' \gg n$ simultaneously.

We do not see any promising way to circumvent the use of the union bound, and it is unclear whether a better upper bound for the quantiles subroutine is possible.

\subsection*{Random Preferences}

For random preferences, similar difficulties to those above again arise; we proceed to outline them, but in slightly less detail.    
Our noiseless upper bound (\Cref{thm:noiseless-comparison-upper-random}) is based on the selection subroutine (\Cref{lem:selection}).  
In this subroutine, the goal is to find the top $\ell$ items among $m'$ items (in generic notation), and in the noiseless case the query complexity is $O(m')$.  
One may hope that this complexity becomes $O\big(m' \log \frac{1}{\delta_0}\big)$ for success probability $1-\delta_0$ in the presence of noise.  
However, the query complexity of this problem is known to be $O\big(m' \log\frac{\ell}{\delta_0} \big)$ \cite{FeigeRaPe94}. 
In our noiseless algorithm, we initially choose $\ell$ to be $m/n$, which immediately suggests a $\log\frac{m/n}{\delta_0}$ factor in our final bound.  
Unfortunately, any benefit of the division by $n$ will be lost after a union bound over the $n$ agents.  
Again, we do not see any apparent way to circumvent this, whether it be by setting $\ell$ differently or avoiding the union bound.

\end{document}